\documentclass[a4paper,10pt]{article}
\usepackage[T1]{fontenc}
\usepackage[english]{babel}
\usepackage{amssymb,amsmath,latexsym,epsfig}
\usepackage[latin1]{inputenc}
\usepackage[T1]{fontenc}
\usepackage{bm}

\makeatletter
\newcommand\ackname{Acknowledgements}
\if@titlepage
  \newenvironment{acknowledgements}{%
      \titlepage
      \null\vfil
      \@beginparpenalty\@lowpenalty
      \begin{center}%
        \bfseries \ackname
        \@endparpenalty\@M
      \end{center}}%
     {\par\vfil\null\endtitlepage}
\else
  \newenvironment{acknowledgements}{%
      \if@twocolumn
        \section*{\abstractname}%
      \else
        \small
        \begin{center}%
          {\bfseries \ackname\vspace{-.5em}\vspace{\z@}}%
        \end{center}%
        \quotation
      \fi}
      {\if@twocolumn\else\endquotation\fi}
\fi
\makeatother

\title{A generalization of Kung's theorem\footnote{The original publication is available at http://link.springer.com/article/10.1007/s10623-015-0139-6}}
\author{Trygve Johnsen\thanks{ Dept. of Mathematics, University of Troms{\o}, N-9037 Troms{\o}, Norway, \texttt{Trygve.Johnsen@uit.no}} \and Keisuke Shiromoto\thanks{Dept. of Mathematics and Engineering, Kumamoto University, 2-39-1 Kurokami, Kumamoto, 860-8555 Japan, \texttt{keisuke@kumamoto-u.ac.jp}} \and Hugues Verdure\thanks{ Dept. of Mathematics, University of Troms{\o}, N-9037 Troms{\o}, Norway, \texttt{Hugues.Verdure@uit.no}} }

\newtheorem{definition}{Definition}[section]
\newtheorem{remark}{Remark}[section]
\newtheorem{lemma}{Lemma}[section]
\newtheorem{corollary}{Corollary}[section]
\newtheorem{proposition}{Proposition}[section]
\newtheorem{theorem}{Theorem}[section]
\newtheorem{example}{Example}[section]

\newenvironment{proof}[1][Proof]{\begin{trivlist}
\item[\hskip \labelsep {\bfseries #1}]}{\end{trivlist}}

\newcommand{\qed}{\nobreak \ifvmode \relax \else
      \ifdim\lastskip<1.5em \hskip-\lastskip
      \hskip1.5em plus0em minus0.5em \fi \nobreak
      \vrule height0.75em width0.5em depth0.25em\fi}

\begin{document}

\maketitle

\begin{abstract}\noindent 
We give a generalization of Kung's theorem on critical exponents of linear codes over a finite field, in terms of sums of extended weight polynomials of linear codes.
For all $i=k+1,\ldots,n$, we give an upper bound on the smallest integer $m$ such that there exist $m$ codewords whose union of supports has cardinality at least $i$.

\end{abstract}

\newcommand{\Proj}{\mathbb{P}}
\newcommand{\N}{\mathbb{N}}
\renewcommand{\k}{\Bbbk}
\newcommand{\Fq}{\mathbb{F}_q}
\newcommand{\K}{\mathbb{K}}
\newcommand{\I}{\mathcal{I}}
\newcommand{\B}{\mathcal{B}}
\newcommand{\F}[1]{\mathbb{F}_{#1}}
\begin{acknowledgements}\noindent The first and third author are grateful to Thomas Britz and the second author for sharing early versions of the  joint  manuscript~\cite{BS} with all of us. The material and results there gave the inspiration for writing the present article.\\ 
The second author was supported by JSPS KAKENHI Grant Number 14474695.\\
The first author wants to thank IMPA, Rio de Janeiro, where he stayed while this work was completed.
\end{acknowledgements}
%
%

\section{Introduction and main result}

For a non-degenerate $[n,k]$ code $C$ over a finite field $\Fq$, let $c_j$ be the smallest integer $m$ such that there exists a codeword of Hamming weight at least $j$ in $C \otimes_{{\mathbb F}_q}{\mathbb F}_{q^m}$. In~\cite{Jthesis,JP}, the authors show that there exist polynomials $P_j(x)$, only dependent on the matroid associated to (any parity check matrix of) $C$, such that the number of codewords of Hamming weight $j$ in $C \otimes_{{\mathbb F}_q}{\mathbb F}_{q^m}$ is $P_j(q^m)$ (see Theorems 5.4 and 5.9 of~\cite{JP} for example).

We now recall the definition of these polynomials:
\begin{definition} \label{polynomials}
Let $M_C$ be the matroid associated to the parity check matrices of $C$, and $n_{M_C}$ its nullity function.
Set 

\[P_{j}(x)=(-1)^j\sum_{|\sigma|=j}\sum_{\gamma\subseteq\sigma}(-1)^{|\gamma|}x^{n_{M_C}(\gamma)} \text{ for }0\leqslant j\leqslant n.\] 
\end{definition} 

Hence we see that $c_j$ is the smallest integer $m$ such that $$P_j(q^m)+\ldots+P_n(q^m) \ne 0.$$
In the special case $j=n$, we see that $c_n$ is the smallest $m$ such that $P_n(q^m) \ne 0$.
The polynomials $P_j(x)$ can of course also be expressed in terms of functions obtained from the matroid of generator (instead of parity check) matrices of $C$. Then $$P_n(x)= \sum_{A \subset E}(-1)^{|A|}x^{\rho(E)-\rho(A)},$$ where $\rho$ is the rank function associated to  the matroid of generator matrices of $C$, and $E=\{1,2,\ldots,n\}.$ For simplicity, and in order to be in harmony with established notation, we will denote  $P_n(x)$ by $p(M_C; x)$.

This makes it natural to quote the "Critical Theorem" by Crapo and Rota,~\cite{CR}:
\begin{theorem}
Let $C$ be an $[n,k]$ code over $\Fq$. For any $X \subseteq E$ and any natural number $m$, the number of ordered $m$-tuples $({\bm v}_1, \ldots,{\bm v}_m)$ of codewords in $C$ with $\mbox{\rm supp}({\bm v}_1) \cup \ldots \cup \mbox{\rm supp}({\bm v}_m) = X$ is $p(M_{C/(E-X)};q^m),$ where $C/(E-X)$ is the code shortened by $E-X$. 
\end{theorem}

\begin{remark} \label{key}
In particular the number of ordered $m$-tuples 
$({\bm v}_1, \ldots,{\bm v}_m)$ of codewords in $C$ with $\mbox{\rm supp}({\bm v}_1) \cup \ldots \cup \mbox{\rm supp}({\bm v}_m) = E$ is $p(M_C;q^m).$ 
Hence we see that the number of ordered $m$-tuples of codewords from $\Fq$ with support $E$, as described in~\cite{CR}, is the same as the number of codewords over ${\mathbb F}_{q^m}$ (with the same generator matrix) with support $E$, as counted above. The analogous statement
holds, also for any subset $X$ of $E$. This follows, for example, from Proposition 4 and Corollary 5 of~\cite{J}. In Proposition 4 one sets up an isomorphism between the set of ordered $m$-tuples over $\Fq$ in question, and the set of codewords over ${\mathbb F}_{q^m}$ in question, and in
Corollary 5 one states that this isomorphism is support-preserving (since the union of the supports of $m$ codewords over $\Fq$ is the same as the support of the subspace (of dimension at most $m$ over $\Fq$) generated by these words). 
Hence the formula in the Critical Theorem also counts the number of codewords over ${\mathbb F}_{q^m}$ with support $X$, and we understand that
for each $j$, the number $P_j(q^m)$ described above, also counts the number of ordered $m$-tuples of codewords from $\Fq$ with support weight $j$, that is, the size of their union of supports. Moreover this number appears as a sum of contributions $p(M_{C/(E-X)};q^m)$, where the sum is taken over all subsets $X$ of cardinality $j$. These observations give:
\end{remark}

\begin{corollary} \label{keycor}
For $1 \leqslant j \leqslant n$, we have: $c_j$ is the smallest integer $m$ such that there exists $m$ codewords over $\Fq$, whose 
union of supports has cardinality at least $j$.
\end{corollary}

We now give: 
\begin{definition}
Let $C$ be a non-degenerate $[n,k]$ code. For $1 \leqslant j \leqslant n$, let  $\gamma_j$ be the smallest integer $m$ such that there exists a $m$-dimensional (over $\Fq$) subspace $D$ of $C$ with support weight at least $j$.
\end{definition}

We  then have:
\begin{proposition} \label{equal}
For each $1 \leqslant j \leqslant n$, we have $\gamma_j=c_j$.
\end{proposition}
\begin{proof}
The proof for the case $j=n$, as given in Lemma 3 of~\cite{BS}, immediately carries over to any $j$ in question.
\qed\end{proof}

Moreover, if $G$ is any generator matrix of $C$, then we have:

\begin{proposition}
\label{gamma_i}
For $1 \leqslant j \leqslant n$ we have: $c_j=\gamma_j=$ the smallest positive integer $m$ for which there exists a $(k-m)$-dimensional subspace of $\Fq^k$ that contains at most $n-j$ columns of $G$.
\end{proposition}
\begin{proof}
The proof for the case $j=n$, stated in Lemma 5 of~\cite{BS}, and originally in Lemma 2 of~\cite{Kl} carries over to any $j$ in question. 
\qed\end{proof}

Let $C$ be an $[n,k]$ code over ${\mathbb F}_q$.
For any $i$, $1 \leqslant i \leqslant k$, the $i$-{\it th generalized Hamming weight} of $C$ is defined by $$d_i=\min\{|\mbox{\rm Supp}(D)| \: : \: \mbox{\rm $D$ is an $[n,i]$ subcode of $C$}\},$$ where $\mbox{\rm Supp}(D)=\bigcup_{{\bm x}\in C}\mbox{\rm supp}({\bm x})$. Then the following result is known as the {\it generalized Singleton bound} (cf.~\cite{W}).
\begin{lemma} \label{singleton}
\label{GSbound}$$d_i \leqslant n-k+i,\:\: i=1,2,\ldots,k.$$
\end{lemma}

For $1 \leqslant i \leqslant n-k$, let $s_i^\perp$ be the $i$-th Singleton defect $n-(n-k)+i-d_i^{\perp} =k+i-d_i^{\perp}$, where $d_i^{\perp}$ is the $i$-th generalized Hamming weight of the dual code $C^{\perp}$. We now give the main result of the paper.

\begin{theorem}[Main result]\label{main}
For a non-degenerate $[n,k]$ code $C$ over $\Fq$ with $n >k$, we have $$c_i=\gamma_i \leqslant s_{n+1-i}^\perp+2$$ for $i=k+1,\ldots, n.$
\end{theorem}

\begin{remark}\label{bad_indices}
We obviously have $\gamma_{i} = 1$ for all $i \leqslant k$, since the word $\Sigma$, being the sum of all rows in any generator matrix $G$, on standard form after a possible permutation of columns, has Hamming weight at least $k$.

If we define $s_i^\perp=-1$ for $n-k+1 \leqslant i\leqslant n$, then the theorem remains true for $1\leqslant i \leqslant k$, and the bound is actually sharp for these $i$. We use this convention in the sequel.
\end{remark}

\begin{corollary}[Kung's theorem,~\cite{K}]\label{main2}
For a non-degenerate $[n,k]$ code $C$ over $\Fq$, we have $$\gamma_n \leqslant k-d^{\perp}+3,$$ where $d^{\perp}$ is the minimum Hamming weight of the dual code $C^{\perp}$.
\end{corollary}

\begin{proof}(of the corollary):
From the definition of the $i$-th Singleton defect, we have that \begin{eqnarray*} s_{(n+1)-n}^\perp+2 &=& s_1^\perp+2\\&=& (k+1-d_1^{\perp})+2\\ &=& k-d_1^{\perp}+3\\&=& k-d^{\perp}+3.\end{eqnarray*}\qed
\end{proof}

\begin{remark}
{In~\cite{JRV}, one expresses the polynomials from the present Definition~\ref{polynomials} in terms of the Betti numbers of the Stanley-Reisner rings of the matroid associated to the parity check matrix of the code in question, and so-called elongations of this matroid. The matroids are then thought of simplicial complexes with facets the bases of the matroids. Our main result is formulated in terms of properties of these polynomials. Hence the material in the present paper gives another link between coding theory properties, and algebraic topological properties of associated matroids appearing in a natural way.}
\end{remark}

\section{Proof of the generalization of Kung's bound}

We now prove the main result of the article.

\begin{proof} (of Theorem~\ref{main}): Let $G$ be a generator matrix of $C$. Set $d_{n+1-i}^{\perp}:=t$. Then it is known that the rank of any set of $(t-1)$ column vectors in $G$ is at least $t-(n+1-i)$ and there exists a set of $t$ column vectors whose rank is $t-(n+1-i)$ (see, for instance~\cite{W}). We assume that $C$ has a generator matrix of the form $[I_k|A]$. From the generalized Singleton bound, it follows that $$ t \leqslant k+(n+1-i).$$ So we can take ${\bm e}_1,\ldots, {\bm e}_{t-n-2+i}$ as linearly independent column vectors in $G$, where ${\bm e}_i$ denotes the vector in $\mathbb{F}_q^k$ with a $1$ in the $i$th coordinate and $0$'s elsewhere.

Assume, for contradiction, that $c_i=\gamma_i > s^{\perp}_{n+1-i}+2=k+3+n-i-t$. By Proposition~\ref{gamma_i}, this means that any $(k-(n+k-i-t+3))=(t+i-3-n)$-dimensional subspace of $\Fq^k$ contains at least $n-i+1$ column vectors in $G$. Let $\alpha \in \Fq$ be non-zero. Then the $(t-n-3+i)$-dimensional subspace $D$ generated by ${\bm e}_1+\alpha{\bm e}_{t-n-2+i},\ldots,{\bm e}_{t-n-3+i}+\alpha{\bm e}_{t-n-2+i}$, contains (at least) $n-i+1$ distinct column vectors of $G$, say \begin{eqnarray*}   {\bm u}_1 &=&(u_1^{(1)}, \ldots,u_{t-n-2+i}^{(1)},0,\ldots,0)^t\\         &\vdots&\,\\  {\bm u}_{n-i+1} &=&(u_1^{({n-i+1})}, \ldots,u_{t-n-2+i}^{({n-i+1})},0,\ldots,0)^t.\,\end{eqnarray*} None of these ${\bm u}_j$ is equal to any ${\bm e}_l$, for $j=1,\ldots,n-i+1$ and $l=1,\ldots,t-n-2+i.$ (Set a linear combination of ${\bm e}_1+\alpha{\bm e}_{t-n-2+i},\ldots,{\bm e}_{t-n-3+i}+\alpha{\bm e}_{t-n-2+i}$ equal to $e_l$, and derive the contradiction $\alpha=0$.)

Therefore $\{{\bm u}_1,\ldots,{\bm u}_{n-i+1},{\bm e}_1,\ldots,{\bm e}_{t-n-2+i}\}$ is a set of $(t-1)$ column vectors in $G$. We see any $t-n-1+i$ vectors of this set are linearly dependent, since they effectively are $t-n-1+i$ vectors in $(t-n-2+i)$-space. Hence the rank of the subspace of $\Fq^k$ spanned by these $t-1$ vectors is strictly less than $t-(n+1-i)$. 

This is a contradiction, and the proof is complete.
\qed\end{proof}

\section{An analysis of the sharpness of the new bound}
Some initial observations are:

Inserting $i=k+1$ in the statement of Theorem~\ref{main}, we get $$\gamma_{k+1}= \gamma_ {n+1-(n-k)}\leqslant s_{n-k}^\perp +2.$$ If $C^{\perp}$ is also non-degenerate, we have  $s_{n-k}^\perp=0$, so $\gamma_{k+1} \leqslant 2$.

It is well known that $s_1^\perp \geqslant s_2^\perp \ldots \geqslant s_{n-k}^\perp$ for any linear code (cf.~\cite{W}), and it is clear from the definitions that $\gamma_i \leqslant \gamma_j$ if $i \leqslant j$.  

We observe that for MDS codes, that is, an $[n,k]$ code over $\Fq$ with $d_1=d=n-k+1$, the theorem gives $\gamma_i \leqslant 2,$ for all $i \geqslant k+1$, but it is clear that there exist MDS codes with $n >k$ and individual codewords of Hamming weight $n$, and then $\gamma_i=1$ for all $i$. Hence we do not claim that the bound given in the theorem is sharp, neither just for $i=n$, nor stronger: for all $i=k+1,\ldots,n$, simultaneously.

In~\cite{BS} one finds, to the contrary:  
\begin{theorem}[\cite{BS}, Theorem 17] Let $C$ be a non-trivial $[n,k]$ code over $\Fq$, with $q$ odd, and $d^{\perp}> 3$. Then $(c_n =)\gamma_n \leqslant k - d^{\perp} +2$.
\end{theorem}

One expects the result to be true also for $d^{\perp}=3$ and $q$ odd, with the extra assumption that $C$ is \underline{not} a simplex code. (We recall that a simplex code is the dual of a Hamming code, and that for these codes, $d^{\perp}= 3$.) Also for binary codes this result holds (see~\cite{BS}, Corollary 15, for a statement also involving most cases with  $d^{\perp}=3$ ).

Hence, at least for odd $q$, and $q=2$, there are almost no linear codes with simultaneous sharpness for all $j\in \{k+1,\ldots, n\}$ in Theorem~\ref{main}, and $d^{\perp} \geqslant 3$, since there are almost no codes where the bound is sharp just for $j=n$. On the other hand, as pointed out in~\cite{BS} (see also Corollary~\ref{mds} below):

\begin{proposition}
If an $[n,k]$ code  $C$ over $\Fq$ is a simplex code, or an MDS code having no codewords of Hamming weight $n$, then 
$(c_n =) \gamma_n = k - d^{\perp} +3$.
\end{proposition}

On the other hand again, if $C$ is an MDS code having no codewords of weight $n$, it is easily shown, as in~\cite{BS}, that  $\gamma_{n-1}=1$, so the only possibility for simultaneous sharpness for all $j \geqslant k+1$ is $k+1=n$, in other words an $[n,n-1]$ code. We will now, after a suitable definition valid for linear codes in general, take a closer look at simplex codes. We will study these codes in particular,  since they serve as a demonstration of how good the bound in Theorem~\ref{main} after all is. 
   
\begin{definition}
For a non-degenerate $[n,k]$ code $C$ over $\Fq$ and $i=k+1,\ldots, n$, we set $t_i=t_i(C)=  s_{n+1-i}^\perp+2 - \gamma_i. $

\end{definition}

One understands that $t_i(C)$ measures the lack of sharpness of the bound in Theorem~\ref{main} for each $i$.

Let ${\cal H}^{\perp}$ be a $[(q^k-1)/(q-1), k, q^{k-1}]$ simplex code over ${\mathbb F}_q$. We set $\mu_s:=(q^s-1)/(q-1)$. Then, for each $r = 1,\ldots,k$, the number of $r$-dimensional subcodes of ${\cal H}^{\perp}$ of support weight $i$ is given by \[    A_{i}^{(r)}({\cal H}^{\perp})  = \begin{cases}     {k \brack r}_q & ,\, i = \mu_k-\mu_{k-r};\\       0            & ,\, \text{\rm otherwise}    \end{cases}\] where ${k \brack r}_q$ denotes the Gaussian binomial coefficient (cf.~\cite{Kl}). In particular, all the subcodes of dimension $r$ have support weight $d_r=\mu_k-\mu_{k-r}$. Therefore we can easily derive the following result.

\begin{lemma}\label{lem:klove92}
For any $r$ such that $1\leqslant r \leqslant k$, $\gamma_i=r$ if and only if $\mu_{k}-\mu_{k-r+1}< i \leqslant \mu_{k}-\mu_{k-r}$.
\end{lemma}

\begin{proof} Indeed, if $\mu_{k}-\mu_{k-r+1}< i$, then there is no subcode of dimension at most $r-1$ with weight support at least $i$, while for $i \leqslant \mu_{k}-\mu_{k-r}$, all subcodes of dimension $r$ have weight support at least $i$.
\qed\end{proof}

Moreover the weight hierarchy of a Hamming $[n=(q^k-1)/(q-1), n-k, 3]$ code ${\cal H}$ over ${\mathbb F}_q$ is determined as follows (cf.~\cite{W}):\begin{eqnarray*}\{d_r({\cal H})\: : \: 1\leqslant r \leqslant n-k\}&=&\{1,2,\ldots,n\} \setminus \{n+1-(\mu_k - \mu_{k-i})\: : \: 1 \leqslant i \leqslant k\} \\&=&\{1,2,\ldots,n\}\setminus \{\mu_{k-i}+1 \: : \: 1\leqslant i \leqslant k\}.\end{eqnarray*}Then it follows that:

\begin{lemma}\label{lem:wei91}
Let $1\leqslant r \leqslant n-k$. Then there exists a unique $2 \leqslant j \leqslant k$ such that $\mu_{j-1}-j+2 \leqslant r \leqslant \mu_j-j$. Moreover, \[d_r({\cal H}) = r+j.\]
\end{lemma}

\begin{proof}The first part is obvious, while the second part is a direct consequence of the description of the set $\{d_r({\cal H})\: : \: 1\leqslant r \leqslant n-k\}$ above. \qed \end{proof}

\begin{proposition} \label{smallt}
Let ${\cal H}^{\perp}$ be a $[(q^k-1)/(q-1), k, q^{k-1}]$ simplex code over ${\mathbb F}_q$. For each $r$, $1\leqslant r \leqslant k$, if $i=\mu_k-\mu_{k-r+1}+l$, $l=1,2,\ldots,q^{k-r}$, then \[t_i=\begin{cases} 0 & (l\leqslant k-r+1)\\ 1 & (k-r+2 \leqslant l). \end{cases}\]
\end{proposition}

\begin{proof}
From Lemma~\ref{lem:klove92}, we have that $\gamma_i=r$. We set $$s:=n+1-i=\mu_k+1-(\mu_k-\mu_{k-r+1}+l)=\mu_{k-r+1}+1-l.$$ In the case of $r >1$ and $l\leqslant k-r+1$, it follows that $$ \mu_{k-r+1}-(k-r+2)+2 \leqslant s (\leqslant \mu_{k-r+2}-{k-r+2}),$$ and so $d_s({\cal H})=s+(k-r+2)$ by Lemma~\ref{lem:wei91}. Therefore we have that \begin{eqnarray*}t_i&=&(k+(n+1-i)-d_s({\cal H})+2)-\gamma_i\\&=& k+\mu_k+1-(\mu_k-\mu_{k-r+1}+l)-(\mu_{k-r+1}+1-l+k-r+2)+2-r\\&=&0.\end{eqnarray*} When $r=1$ and $l\leqslant k-r+1$, it follows from our convention in Remark~\ref{bad_indices} that in that case too, \[t_i = s_{n+1-i}^\perp + 2 - \gamma_i = -1 + 2 -1 = 0.\]

Similarly, in the case of $k-r+2 \leqslant l$, it follows that $d_s({\cal H})=s+(k-r+1)$. Thus we have that \begin{eqnarray*}t_i&=&(k+(n+1-i)-d_s({\cal H})+2)-\gamma_i\\&=& k+\mu_k+1-(\mu_k-\mu_{k-r+1}+l)-(\mu_{k-r+1}+1-l+k-r+1)+2-r\\&=&1.\end{eqnarray*} The proposition follows.
\qed\end{proof}

\begin{example} \label{Hamming}
As an illustration, let us study the $[(q^k-1)/(q-1), k, q^{k-1}]$ simplex code for $q=2$ and $k=4$. Here $n=15, k=4$, and $(d_1,d_2,d_3,d_4)=(8,12,14,15)$. This is a constant weight code, and then $\gamma_i= \min \{j | d_j \geqslant i \}$ for all $i$. 
The dual code has the following weight hierarchy:  \[(d_1^{\perp},\ldots,d_{11}^{\perp})=(3,5,6,7,9, 10,11,12,13,14,15).\]
We obtain the following table, where we only list the interesting values $i=k+1=5, \ldots, n=15$. (Trivially $\gamma_i=1$ for $i=1,\ldots,4$ and strictly speaking,  $s_{n+1-i}^\perp$ is not defined for these values of $i$. See Remark~\ref{bad_indices}.)

\begin{tabular}{|c||c|c|c|c|c|c|c|c|c|c|c|}
\hline
$i$ &   $5$ & $6$ & $7$ & $8$ & $9$ & $10$ & $11$ & $12$ & $13$ & $14$ & $15$ \\ \hline
$s_{16-i}^\perp+2$&   2 & 2 & 2 & 2 & 2 & 2 & 2 & 3 & 3 & 3 & 4 \\ \hline
$c_i=\gamma_i$ &   1 & 1 & 1 & 1 & 2 & 2 & 2 & 2 & 3 & 3 & 4 \\ \hline
$t_i$ &  1 & 1 & 1 & 1 & 0 & 0 & 0 & 1 & 0 & 0 & 0 \\ \hline
\end{tabular} 
\end{example}

\begin{remark} \label{smalltalk}
We observe that even if our bound is not sharp in the example above, nor in any other simplex codes  (Proposition~\ref{smallt}), it is reasonably good; the $t_i$ are $0$ or $1$, and $0$  "quite often". 

But as one sees from  Proposition~\ref{smallt}, for all non-trivial cases of simplex codes there will be \underline{some} $i$, with $k+1 \leqslant i <n$ with $t_i =1$, (see the case $i=12$ in Example~\ref{Hamming} above), so the generalization of Kung's bound in Theorem~\ref{main} will never be sharp for all $j$ in the range $[k+1,n]$, even for simplex codes.
  
On the other hand, we see that $t_{n+1-j}=0$, so that our bound is sharp, for all $j=1,q,q+1,q^2+q-1,q^2+q,q^2+q+1, q^3+q^2+q-2, q^3+q^2+q-1,q^3+q^2+q,q^3+q^2+q+1, \ldots,$ with $j \leqslant n-k.$ The smallest integer $j$ not of this form, for any prime power $q$, is $22$. For MDS codes, it is trivially true that the $t_i$ are  $0$ or $1$ since all $s^\perp_i$ in question are zero by definition, for $i=1,\ldots,n-k,$ and  the corresponding $\gamma_{n+1-i}$ are $1$ or $2$.
\end{remark}

\begin{example} \label{ReedMuller}
Quite "often" a linear code has a codeword of weight $n$, and then $c_i=1$ for all $i$, and then there is of course no need to apply the bound in our Theorem~\ref{main} to determine the $c_i$. (See~\cite{BB} where the authors discuss when a linear code might have a codeword of maximum weight.) So in such cases that bound is expected to be very unsharp, and our $t_i$ fairly big. As a contrast to the examples above, where we looked at codes that are more "sparse" when it comes to non-zero elements, even for those codewords that have largest weight, we here give an  example of a class of codes with codewords of weight $n$.

Let $C$ be the Reed-Muller code of first order and rank $r$ over $\Fq$. This is a $[q^{r-1},r,q^{r-2}(q-1)]$ linear code. These codes are often given as examples alongside with simplex codes in articles and textbooks. But in the context we discuss here these classes are opposite extremes. The word $(1, \ldots,1)$ is a codeword of the Reed-Muller code, and therefore $c_i = 1,$ for $ 1 \leqslant i \leqslant n=q^{r-1},$ for all $r$. The weight hierarchy of $C$ is given by \[d_i = \left \{ \begin{array}{ll} q^{r-1}-q^{r-i-1} & \mathrm{ if }\: 1 \leqslant i \leqslant r-1 \\ q^{r-1} & \mathrm{ if }\: i=r \end{array}\right.\] By Wei duality, we can therefore find the weight hierarchy of the dual code. Namely, \begin{eqnarray*}(d^\perp_1,\ldots,d^\perp_{q^{r-1}-r}) &=& \{1,\ldots,q^{r-1}\} \setminus\{q^{r-1}+1-d_i\: : \: 1 \leqslant i \leqslant r\} \\ &=& \{2,\ldots,q^{r-1}\} \setminus(\{q^{r-i-1}+1\: : \: 1 \leqslant i \leqslant r-1\} \end{eqnarray*} From there, we get that \[s^\perp_{n+1-j} = k+n+1-j-d^\perp_{n+1-j}= i-1 \Leftrightarrow q^{r-1}-q^{r-i}+r-i+2 \leqslant j \leqslant q^{r-1}-q^{r-i-1}+r-i.\] We can then evaluate the sharpness of our bound and we get that \[t_j = i \Leftrightarrow q^{r-1}-q^{r-i}+r-i+2 \leqslant j \leqslant q^{r-1}-q^{r-i-1}+r-i.\] In particular \[t_{q^{r-1}} = r-1.\]
\end{example}

An $[n,k]$ code $C$ over ${\mathbb F}_q$ is called an $r$-{\it{MDS code}} iff it attains the generalized Singleton bound $d_r = n-k+r$ from Lemma~\ref{singleton} for $i=r$. Note that if a code is $r$-MDS, then it also $i$-MDS for all $r<i$, and that a code is MDS if and only if it is $1$-MDS. We denote the maximum Hamming weight of a single non-zero codeword in $C$ by $d_{\mbox{\rm max}}$.

\begin{proposition} \label{hei} Let $C$ be an $[n,k]$ code over ${\mathbb F}_q$. If $C^{\perp}$ is an $(n+1-i)$-MDS code and $d_{\mbox{\rm max}}<i$, then $C$ attains the bound for $c_i$ in Theorem $\ref{main}$.
\end{proposition}

\begin{proof}
Since $C^{\perp}$ is an $(n+1-i)$-MDS code, we have that the bound in Theorem~\ref{main} is: $$ c_i\leqslant (k+n+1-i-d_{n+1-i}^{\perp})+2=2. $$ Moreover $C$ does not contain any codeword whose Hamming weight is greater than $i-1$. Thus we find that $c_i \geqslant 2 $ also, and so the proposition holds  (with bound $2$ in this case).
\qed\end{proof}

For the special case $d=n-k+1$ and  $i=n$, we obtain:
\begin{corollary} \label{mds}
Let $C$ be an MDS code over ${\mathbb F}_q$. If $d_{\mbox{\rm max}}<n$, then the bound in Theorem $\ref{main}$ is sharp for $i=n$ (with $c_n=\gamma_n=2$).
\end{corollary}

\begin{example} \label{sparse}
Another example of sharpness and the lack of it is provided a binary code given by the parity check matrix \[H=\begin{bmatrix}1 & 1 & 1 & 0 & 0 & 0 & 0 & 0 & 0 \\
0 & 0 & 0 & 1 & 1 & 1 & 0 & 0 & 0 \\
0 & 0 & 0 & 0 & 0 & 0 & 1 & 1 & 1\end{bmatrix}.\]Here $n=9$, and $k=6$, and the $d_j^{\perp}$ are $3,6,9$, so that $s_1^{\perp}=4,$ $s_2^{\perp}=2$, and $s_3^{\perp}=0$. Moreover we have $c_7=c_8=c_9=2$ since $d_{\mbox{\rm max}}=6$ (and one easily finds two words whose support union is $\{1,2,\ldots,9\}$). The upper bounds for the $c_i$ given by Theorem~\ref{main} are $\{2,4,6\}.$  So this is sharp for $c_i$ only when $i=7$ (in the "proper" range $[k+1,n]=[7,9]$). We observe that Proposition~\ref{hei} gives sharpness for $i=7$ since  $C^{\perp}$ is $n+1-i=9+1-7=3$-MDS, and $d_{\mbox{\rm max}} < 7.$

More generally, let $H$ be a $m \times 3m$ parity check matrix with $3$ consecutive $1$'s grouped together "on the diagonal" in an analogous pattern. It defines a $[3m,2m]$ binary code $C$. It is easy to see that $d_{\mbox{\rm max}}=2m$ and that we can cover the whole space with two codewords. Thus \[c_{2m+1}=\ldots=c_{3m}=2.\] We also have that the weight hierarchy of $C^\perp$ is \[(d^\perp_1,\ldots,d^\perp_m)=(3,6,\ldots,3m)\] and that \[(s^\perp_1,\ldots,s^\perp_m)=(2m-2,2m-4,\ldots,0).\] Then the bound is sharp for $i=2m+1$ since \[c_{2m+1}=2 = s^\perp_{3m+1-(2m+1)} + 2,\] while it is non-sharp for $i\geqslant 2m+2$. The sharpness of the bound for $i=2m+1$ can be derived from Proposition~\ref{hei} since $2m=d_{\mbox{\rm max}} < 2m+1$ and $C$ is $(3m+1-(2m+1))$-MDS.\\

There is an obvious further generalization of this example to the case of $H$ with $m$ rows with $2l+1$  consecutive $1$'s grouped  together ``on the diagonal'' in an analogous pattern, for any natural number $l$.
\end{example} 

In Corollary~\ref{mds} and Example~\ref{sparse} we show sharpness of the bound in Theorem $\ref{main}$  for the "extreme cases" $i=n$ or $i=k+1$. Using the following corollary of Proposition~\ref{hei} we finish with an example where we obtain sharpness  for some intermediate values of $i$ in many cases.

\begin{corollary} \label{smart}
Let $C$ be an $[n,k,d]$ code over ${\mathbb F}_q$. If $d_{\mbox{\rm max}}<k+d-1$, then  the bound in Theorem $\ref{main}$ is sharp for $i=k+d-1$.
\end{corollary}

\begin{proof}
Set $s:=n+2-k-d=n+1-i$. From the generalized Singleton bound, it follows that $$d_s^{\perp}\leqslant k+(n+2-k-d)=n+2-d.$$ By  Wei  duality, this is in fact an equality, because of  $$(d_s^{\perp},d_{s+1}^{\perp},\ldots,d_{n-k}^{\perp})=(n+2-d,n+3-d,\ldots,n).$$ Thus $C^{\perp}$ is $n+1-i=s$-MDS. Thus the result follows from Proposition~\ref{hei} (and $c_i=\gamma_i=2$).
\qed\end{proof}

\begin{example} \label{multiplyprojective}
Let $k$ and $j$ be natural numbers, and let $G$ be a $k\times n$ matrix over ${\mathbb F}_q$, where $n=j\frac{q^k-1}{q-1},$ and where we for each point of ${\mathbb P}_n^{k-1}$ over ${\mathbb F}_q$  use $j$ of its representatives in ${\mathbb F}_q^k$ as columns of $G$. Using $G$ as generator matrix of a $q$-ary code $C$ we get a simplex code for $j=1$. For  an arbitrary natural number $j$ we get a constant weight code with weight $d(C)=d_1(C)=d_{\mbox{\rm max}}(C)=jq^{k-1}.$ We assume $k \geqslant 2$ and conclude that $d_{\mbox{\rm max}}(C) < jq^{k-1}+k-1=k+d(C)-1$. Hence Corollary~\ref{smart} can be applied, and our bound for $c_i$ in Theorem~\ref{main} is sharp for $i=jq^{k-1}+k-1.$ From the proof of that result we also know that that bound is $2$. In the case $j=4$, $k=q=2$, we have $k+1=3$, $n=12$, while we obtain sharpness for the intermediate value $i=9$ (in the table of Example~\ref{Hamming} the value for the $i$ in question is $11$).
 \end{example}

\end{document}